\documentclass[year=25,pdfa]{fmcad}
\usepackage{cite}
\usepackage{amsmath,amssymb,amsfonts}
\usepackage{algorithmic}
\usepackage{graphicx}
\usepackage{textcomp}
\usepackage{xcolor}
\usepackage{balance}

\def\BibTeX{{\rm B\kern-.05em{\sc i\kern-.025em b}\kern-.08em
    T\kern-.1667em\lower.7ex\hbox{E}\kern-.125emX}}

\usepackage{dsfont}
\usepackage{amsthm}
\usepackage{cleveref}
\usepackage{thm-restate}
\usepackage{subcaption}
\usepackage{bussproofs}
\usepackage{enumitem}
\usepackage{thm-restate}

\usepackage{tikz}
\usetikzlibrary{shapes,snakes}

\newtheoremstyle{mystyle}%
{}{}%
{\itshape}%
{0pt}%
{\bfseries}
{.}%
{.5em}%
{\indent\thmname{#1}\thmnumber{ #2}\thmnote{ (#3)}}

\declaretheorem[style=mystyle]{lemma}

\declaretheorem[style=mystyle]{example}
\declaretheorem[style=mystyle]{remark}

\declaretheorem[style=mystyle]{definition}

\newcommand{\ap}{\mathit{AP}}
\newcommand{\pathVars}{\mathcal{V}}

\newcommand{\directions}{\mathbb{D}}

\newcommand{\ldot}{\mathpunct{.}}

\newcommand{\quant}{\mathds{Q}}

\newcommand{\nat}{\mathbb{N}}

\newcommand{\calG}{\mathcal{G}}
\newcommand{\calA}{\mathcal{A}}

\newcommand{\calL}{\mathcal{L}}

\newcommand{\calK}{\mathcal{K}}

\providecommand{\ltlN}{\operatorname{%
		\protect\tikz[baseline]{
			\draw[line width=.12ex]
			(0,.6ex) circle (.8ex);
}}}{}

\providecommand{\ltlF}{\operatorname{%
		\protect\tikz[baseline]{
			\draw[line width=.12ex,line join=round]
			(0ex,.6ex) -- (.95ex,1.55ex) -- (1.9ex,.6ex) -- (.95ex,-.35ex) -- cycle;
}}}{}

\providecommand{\ltlG}{\operatorname{%
		\protect\tikz[baseline]{
			\draw[line width=.12ex,line join=round]
			(0ex,-.2ex) -- (0ex,1.3ex) -- (1.5ex,1.3ex) -- (1.5ex,.-.2ex) -- cycle;
}}}{}

\DeclareMathOperator{\ltlU}{\mathcal{U}}

\newcommand{\agents}{\mathbb{P}}
\newcommand{\paths}{\mathit{Paths}}
\newcommand{\traces}{\mathit{Traces}}

\newcommand{\game}[2]{\mathcal{G}_{#1, #2}}
\newcommand{\gamee}[2]{\mathcal{G}^{\forall\exists}_{#1, #2}}
\newcommand{\gamenode}[1]{\langle #1 \rangle}
\newcommand{\winss}[2]{\mathit{wins}(#1, #2)}
\newcommand{\veri}{\mathfrak{V}}
\newcommand{\refu}{\mathfrak{R}}

\setlist[itemize]{leftmargin=*}

\newif\iffullversion
\fullversiontrue

\newcommand{\ifFull}[2]{\iffullversion#1\else#2\fi}

\begin{document}

\title{On Hyperproperty Verification, Quantifier Alternations, and Games under Partial Information}
\author{
	\IEEEauthorblockN{Raven Beutner \orcid{0000-0001-6234-5651}}
\IEEEauthorblockA{\textit{CISPA Helmholtz Center for} \\ \textit{Information Security} \\ \textit{Germany}}
\and
\IEEEauthorblockN{Bernd Finkbeiner \orcid{0000-0002-4280-8441}}
\IEEEauthorblockA{\textit{CISPA Helmholtz Center for} \\ \textit{Information Security} \\ \textit{Germany}}
}

\maketitle

\begin{abstract}
Hyperproperties generalize traditional trace properties by relating multiple execution traces rather than reasoning about individual runs in isolation. 
They provide a unified way to express important requirements such as information flow and robustness properties. 
Temporal logics like HyperLTL capture these properties by explicitly quantifying over executions of a system. 
However, many practically relevant hyperproperties involve \emph{quantifier alternations}, a feature that poses substantial challenges for automated verification.
Complete verification methods require a system complementation for each quantifier alternation, making it infeasible in practice.
A cheaper (but incomplete) method interprets the verification of a HyperLTL formula as a two-player game between universal and existential quantifiers. 
The game-based approach is significantly cheaper, facilitates interactive proofs, and allows for easy-to-check certificates of satisfaction. 
It is, however, limited to $\forall^*\exists^*$ properties, leaving important properties out of reach. 
In this paper, we show that we can use games to verify hyperproperties with arbitrary quantifier alternations by utilizing multiplayer games under \emph{partial information}.
While games under partial information are, in general, undecidable, we show that our game is played under hierarchical information and thus falls in a decidable class of games.
We discuss the completeness of the game and study prophecy variables in the setting of partial information.
\end{abstract}

\section{Introduction}\label{sec:intro}

In 2008, Clarkson and Schneider \cite{ClarksonS08} coined the term \emph{hyperproperties} for the rich class of system requirements that relate multiple executions. 
In contrast to \emph{trace properties} -- i.e., properties over individual executions, expressed, e.g., in linear-time temporal logics (LTL) \cite{Pnueli77} -- hyperproperties can express important properties related to information flow, knowledge, and robustness.
As an example, consider a system with secret input $h$, public input $l$, and public output $o$, and assume we want to express that the public behavior does not leak any information about the secret input. 
We cannot express such an information-flow requirement as a trace property in, e.g., LTL; we need to compare \emph{multiple} executions to see if (and how) the secret input impacts the output.
Instead, we can express it as a hyperproperty in HyperLTL \cite{ClarksonFKMRS14}, an extension of LTL with explicit quantification over execution traces.
For example,
\begin{align}\label{eq:od}
	\forall \pi_1. \forall \pi_2. \ltlG (l_{\pi_1} \leftrightarrow l_{\pi_2}) \to \ltlG (o_{\pi_1} \leftrightarrow o_{\pi_2}). \tag{\textsf{OD}}
\end{align}
requires that any pair of executions $\pi_1, \pi_2$ with identical public input also has the same output, i.e., the output is fully determined by the public input and thus cannot possibly leak the secret input (cf.~\emph{observational determinism} \cite{ZdancewicM03}).

While originating in the study of information flow, hyperproperties have since been established as a much more general framework that captures properties from many different areas, including, e.g., knowledge properties in multi-agent systems (MAS) \cite{HoekW02,BeutnerFFM23}.
As an example, consider some MAS with agents $1, \ldots, n$, and some LTL property $\psi$, and assume that we want to verify that there exists at least one execution of the MAS such that agent $i$ knows that $\psi$ holds (e.g., some adversary knowing some secret).
Formally, \emph{knowing} that $\psi$ holds on some execution trace $\pi_1$ means that $\psi$ must hold on all traces $\pi_2$ that are indistinguishable from $\pi_1$ for agent $i$ (cf.~\cite{halpern1986reasoning}), which \emph{is} a hyperproperty. 
We can easily express this property in HyperLTL, as follows
\begin{align}\label{eq:k1}
	\exists \pi_1. \forall \pi_2\ldot \pi_1 \equiv_i \pi_2 \to \psi[\pi_2],\tag{\textsf{K}$_1$}
\end{align}
where we write $\psi[\pi_2]$ to indicate that $\psi$ holds on trace $\pi_2$, and $\pi_1\equiv_i \pi_2$ denotes that executions $\pi_1$ and $\pi_2$ appear indistinguishable under agent $i$'s observations.
Using the flexibility of quantification, we can also express nested knowledge properties. 
For example, we can express that, on some execution, agent $i$ knows that agent $j$ does \emph{not} know whether $\psi$ holds:
\begin{align}\label{eq:k2}
	\begin{split}
		\exists \pi_1. \forall \pi_2. &\exists \pi_3. \exists \pi_4\ldot \pi_1 \equiv_i \pi_2 \to \\
		&\quad\big(   \pi_2 \equiv_j \pi_3 \land \pi_2 \equiv_j \pi_4 \land \psi[\pi_3] \land \neg \psi[\pi_4]\big),
	\end{split}\tag{\textsf{K}$_2$}
\end{align}
i.e., for every trace $\pi_2$ that agent $i$ cannot distinguish from $\pi_1$, there exist two traces $\pi_3, \pi_4$ that agent $j$ cannot distinguish from $\pi_2$, one of which satisfies $\psi$ and one violates $\psi$. 

\paragraph*{Verification of HyperLTL}

In this paper, we study the model-checking problem for HyperLTL (or, more generally, of hyperproperties that can be expressed using HyperLTL-style quantification over system paths or traces \cite{Rabe16,GutsfeldMO20,CoenenFHH19}).
Unsurprisingly, the quantifier prefix of the HyperLTL formula directly impacts the complexity of this verification problem. 
For alternation-free formulas (e.g., \ref{eq:od}), verification can be reduced to the verification of an LTL property on the self-composition of the system \cite{barthe2011secure,FinkbeinerRS15}, which is very efficient. 
Verification gets much more challenging when the formula includes quantifier alternations as used in \ref{eq:k1}, \ref{eq:k2}, and many other HyperLTL formulas studied in the literature \cite{WangNP20,BiewerDFGHHM22,anand2024verification}.
Complete approaches require one system complementation for each alternation in the formula \cite{ClarksonFKMRS14,FinkbeinerRS15}, making it infeasible in practice. 

\paragraph*{Game-based Verification}

A cheaper (but incomplete) verification method for $\forall^*\exists^*$ formulas (i.e., formulas where an arbitrary number of universal quantifiers is followed by an arbitrary number of existential quantifiers) is based on a strategy-based interpretation of existential quantification \cite{CoenenFST19,BeutnerF22}.
The key idea is to interpret the verification of a HyperLTL formula of the form $\forall \pi_1. \exists \pi_2. \psi$ (where $\psi$ is the LTL body) as a \emph{game} between two players. 
A refuter controls the universally quantified trace by moving through a copy of the underlying system, thereby constructing a concrete trace for  $\pi_1$. 
The verifier reacts to the moves by the refuter and moves through a separate copy of the system, thereby producing a concrete trace for $\pi_2$.
The goal of the verifier is to ensure that $\pi_1$ and $\pi_2$, together, satisfy $\psi$. 
We can think of the verifier's strategy as providing a step-wise construction of a concrete witness trace for $\pi_2$ (akin to a Skolem function).
This game-based approach is sound (i.e., if the verifier wins, the hyperproperty is satisfied by the system) and computationally cheaper than complementation-based approaches.
Moreover, the game-based approach also allows for interactive proofs and witnesses of satisfaction.
For example, we can use the game-based framework to let the user construct a strategy \emph{interactively} \cite{CorrensonF25}, allowing verification even in situations where automated techniques do not scale.
Likewise, we can use a winning strategy for the verifier as an (easy-to-check) certificate that the property is satisfied \cite{BeutnerFG24}. 

\paragraph*{Unsoundness Beyond $\forall^*\exists^*$}

The game-based verification approach of Coenen et al.~\cite{CoenenFST19} and its descendants have proven themselves in many situations (cf.~\Cref{sec:relatedWork}).
However, since its inception, the approach has been limited to $\forall^*\exists^*$ properties.
Intuitively, as soon as we consider properties beyond $\forall^*\exists^*$, the step-wise selection of the traces leads to unsoundness, i.e., cases where a winning strategy exists even though the property is violated; we give a concrete instance in \Cref{ex:unsound}. 
Consequently, for properties outside the $\forall^*\exists^*$ fragment, no effective verification approximation exists (cf.~\Cref{sec:relatedWork}), nor does there exist any approach that allows interactive proofs or satisfaction certificates. 

\paragraph*{Partial Information} 

In this paper, we present a novel game-based method that allows us to soundly verify \emph{arbitrary} quantifier structures.
Our key contribution is the observation that we need to reason about \emph{partial information}.
In our game-based encoding, we consider multiple players, each of whom controls a unique trace in the HyperLTL formula.
We then carefully design an observation model for each player to ensure soundness.
Intuitively, our observations ensure that each player controlling some trace $\pi$ can only observe the state sequence from traces that are quantified \emph{before} $\pi$. 
We thus obtain a multiplayer game played under partial information that, if won by a certain group of players, ensures that the formula holds on the given system (\Cref{sec:game-ii}).

\paragraph*{Hierarchical Information}
In general, multiplayer games under incomplete information are undecidable. 
We show that our verification game falls in a well-known class of games that can be solved effectively.
Namely, games where the information of the players is \emph{hierarchical} \cite{BerthonMM17,berwanger2018hierarchical,MaubertM18,BerthonMMRV21,BerwangerCWDH10,ChatterjeeD10,GastinSZ09}. 

\paragraph*{Completeness and Prophecy Variables}

Similar to the $\forall^*\exists^*$ game \cite{CoenenFST19,BeutnerF22}, our game-based approach using partial information is incomplete: In some cases, the property holds, but no winning strategy exists. 
While incomplete for $\forall^*\exists^*$ properties, we show that our approach is complete for $\exists^*\forall^*$ properties (\Cref{sec:ea}).
This gives rise to a sound-and-complete verification method for all properties with at most one quantifier alternation (via complementation), at the cost of introducing partial information.
More generally, we study the use of prophecy variables \cite{AbadiL88,BeutnerF22} in our verification game, allowing a player to peek at the future temporal behavior of other players (\Cref{sec:comp}).

\paragraph*{Applications}

The core contribution that partial information enables game-based verification of arbitrary hyperproperties creates a plethora of new possibilities for hyperproperty verification:
Our game-based view \textbf{(1)} leverages techniques for solving partial information games (cf.~\Cref{sec:relatedWork}) for automated verification; \textbf{(2)} facilitates interactive proof of hyperproperties beyond $\forall^*\exists^*$ by letting the user construct strategies; \textbf{(3)} enables strategies as easy-to-check certificates for satisfaction; and \textbf{(4)} supports prophecies to soundly strengthen the approach, both in automated and interactive verification.

\paragraph*{Full Version}
Proofs of all results can be found in the \ifFull{appendix}{full version \cite{fullVersion}}.

\section{Preliminaries}\label{sec:prelim}

\paragraph*{Kripke Structures}

As the basic system model, we use finite-state Kripke structures.
We assume that $\ap$ is a fixed set of \emph{atomic propositions} (AP). 
A Kripke structure (KS) is a tuple $\calK = (S, s_\mathit{init}, \directions, \kappa, \allowbreak \ell)$, where $S$ is a finite set of states, $s_\mathit{init} \not\in S$ is a dedicated initial state (\emph{not} part of $S$), $\directions$ is a finite set of directions, $\kappa : (S \uplus \{s_\mathit{init}\}) \times \directions \to S$ is the transition function, and $\ell : (S \uplus \{s_\mathit{init}\})  \to 2^\ap$ labels each state with an evaluation of the APs.\footnote{We use slightly unconventional notation in two places:
	Firstly, we assume a dedicated initial state $s_\mathit{init}$, simplifying the addition of prophecies.
	Secondly, we use \emph{directions} to uniquely identify successor states, simplifying our game construction. 
	Traditionally, transition functions in Kripke structures are functions $S \to 2^S \setminus \{\emptyset\}$ that map each state to a non-empty set of successor states. 
	We can easily transform such a transition function into a directed function $S \times \directions \to S$ by using sufficiently many directions.}
A path in $\calK$ is an infinite sequence $\tau \in (S \uplus \{s_\mathit{init}\})^\omega$ such that $\tau(0) = s_\mathit{init}$, and for every $i \in \nat$, there exists some $d \in \directions$ such that $\tau(i+1) = \kappa(\tau(i), d)$.
We define $\paths(\calK)$ as the set of all paths in $\calK$.
Each path $\tau$ denotes an associated trace $\ell(\tau) := \ell(\tau(0))\ell(\tau(1)) \ell(\tau(2))\cdots \in (2^\ap)^\omega$ defined by applying $\ell$ pointwise.
We define $\traces(\calK) := \{\ell(\tau) \mid \tau \in \paths(\calK)\} \subseteq (2^\ap)^\omega$ as the set of all traces generated by $\calK$.

\paragraph*{Linear-Time Temporal Logic}

Linear-time temporal logic (LTL) \cite{Pnueli77} formulas are defined as follows
\begin{align*}
	\psi := a \mid \psi \land \psi \mid \neg \psi \mid \ltlN \psi  \mid \psi \ltlU \psi,
\end{align*}
where $a \in \ap$ is an atomic proposition.
The basic formula $a$ requires that the AP $a$ holds in the current state, $\ltlN \psi$ requires that $\psi$ holds in the \emph{next} step, and $\psi_1 \ltlU \psi_2$ requires that $\psi_1$ holds \emph{until} $\psi_2$ eventually holds. 
We use the usual derived constants and connectives $\mathit{true}, \mathit{false}, \lor, \to, \leftrightarrow$, and the temporal operators  \emph{eventually} $\ltlF \psi := \mathit{true} \ltlU \psi$, and \emph{globally} $\ltlG \psi := \neg \ltlF \neg \psi$.
Given a trace $t \in (2^\ap)^\omega$, we define the semantics of LTL for each time point $i \in \nat$ as follows:
\begin{align*}
	t, i &\models a &\text{iff } \quad&a \in t(i) \\
	t, i&\models \psi_1 \land \psi_2 &\text{iff } \quad &t,i \models \psi_1 \text{ and } t, i \models \psi_2\\
	t, i &\models \neg \psi &\text{iff } \quad &t, i \not\models \psi\\
	t, i &\models \ltlN \psi &\text{iff } \quad &t, i+1 \models \psi\\
	t, i &\models \psi_1 \ltlU \psi_2  &\text{iff } \quad &\exists k \geq i\ldot t, k \models \psi_2 \text{ and } \\
	&&&\quad\quad\forall i \leq j < k\ldot t, j \models \psi_1 \span \span
\end{align*}
We write $t \models_\mathit{LTL} \psi$ if $t$ satisfies $\psi$, i.e., $t, 0 \models \psi$.

\paragraph*{Deterministic Parity Automata}\label{sec:dpa}

A deterministic parity automaton (DPA) is a tuple $\calA = (\Sigma, Q, q_0, \allowbreak \delta, c)$ where $\Sigma$ is a finite alphabet, $Q$ is a finite set of states, $q_0 \in Q$ is an initial state, $\delta : Q \times \Sigma \to Q$ is a transition function, and $c : Q \to \nat$ colors each state with a natural number.
For an infinite word $u \in \Sigma^\omega$, we define $\mathit{run}(\calA, u) \in Q^\omega$ as the unique run of $\calA$ on $u$. 
Formally, $\mathit{run}(\calA, u)$ is the unique word in $Q^\omega$ such that $\mathit{run}(\calA, u)(0) = q_0$ and for every $i \in \nat$, $\mathit{run}(\calA, u)(i+1) = \delta(\mathit{run}(\calA, u)(i), u(i))$.
That is, we start the run in the initial state $q_0$ and progress by following $\calA$'s transition function using the letters from $u$.
The acceptance condition in DPAs is based on the color of the states (as given by $c$).
An infinite run in $Q^\omega$ is accepting if the minimal color that occurs infinitely often is even. 
We write $\calL(\calA) \subseteq \Sigma^\omega$ for the language of the automaton, which consists of all words where the unique run is accepting.
In this paper, we use a parity acceptance condition as DPAs capture every $\omega$-regular property and thus every LTL-expressible property:

\begin{lemma}[\!\!\cite{EsparzaKRS17,Piterman07}]\label{lem:ltl}
	For every LTL formula $\psi$, we can effectively construct a DPA $\calA_\psi = (2^\ap, Q_\psi, q_{0, \psi}, \delta_\psi,c_\psi)$ such that $\calL(\calA_\psi) = \big\{ t \in (2^\ap)^\omega \mid t \models_\mathit{LTL} \psi \big\}$.
\end{lemma}

We emphasize that our construction also applies to other automaton types. 
For example, if the LTL body of our hyperproperty can be expressed as a \emph{deterministic} Büchi (resp.~safety) automaton, our later construction in \Cref{sec:game-ii} yields a Büchi (resp.~safety) game.

\paragraph*{HyperLTL}\label{sec:hyperltlt}

HyperLTL \cite{ClarksonFKMRS14} extends LTL with explicit quantification over (execution) traces of the system.
Let $\pathVars = \{\pi_1, \pi_2, \ldots\}$ be a set of \emph{trace variables}.
HyperLTL formulas are generated by the following grammar
\begin{align*}
	\psi &:= a_\pi \mid \psi \land \psi \mid \neg \psi \mid \ltlN \psi  \mid \psi \ltlU \psi \\
	\varphi &:=\forall \pi \ldot \varphi \mid \exists \pi \ldot \varphi \mid \psi
\end{align*}
where $a \in \ap$ is an atomic proposition, and $\pi \in \pathVars$ is a trace variable. 
Each HyperLTL formula thus has the form $\varphi = \quant_1 \pi_1 \ldots \quant_n \pi_n\ldot \psi$, where $\quant_1, \ldots, \quant_n \in \{\forall, \exists\}$ are quantifiers, $\pi_1, \ldots, \pi_n \in \pathVars$ are trace variables, and $\psi$ is an LTL formula over trace-variable-indexed APs. 
The formula quantifies over traces $\pi_1, \ldots, \pi_n$ in the system (in typical first-order semantics) and requires that the resulting combination of $n$ traces satisfies the temporal requirement expressed by the trace-variable-indexed LTL formula $\psi$. 

A trace assignment is a partial function $\Pi : \pathVars \rightharpoonup (2^\ap)^\omega$ that maps trace variables to traces. 
Given $\Pi$ we can evaluate the LTL body $\psi$ in each time point $i \in \nat$:
\begin{align*}
	\Pi, i &\models a_\pi &\text{iff } &a \in \Pi(\pi)(i) \\
	\Pi, i&\models \psi_1 \land \psi_2 \!\!\!&\text{iff }  &\Pi,i \models \psi_1 \text{ and } \Pi, i \models \psi_2\\
	\Pi, i &\models \neg \psi &\text{iff }  &\Pi, i \not\models \psi\\
	\Pi, i &\models \ltlN \psi &\text{iff }  &\Pi, i+1 \models \psi\\
	\Pi, i &\models \psi_1 \ltlU \psi_2  \!\!\!&\text{iff }  &\exists k \geq i\ldot \Pi, k \models \psi_2 \text{ and } \\
	&&&\quad\quad\forall i \leq j < k\ldot \Pi, j \models \psi_1 \span \span
\end{align*}
Temporal and Boolean operators are evaluated as for LTL.
Whenever we evaluate an indexed AP $a_\pi$, we look at the trace bound to $\pi$ and check if $a$ currently holds on this trace. 
Given a KS $\calK$, the quantifier prefix in HyperLTL then adds traces to the trace assignment in typical first-order fashion:
\begin{align*}
	\Pi &\models_\calK \psi &\text{iff } \quad& \Pi, 0 \models \psi\\
	\Pi &\models_\calK \forall \pi \ldot \varphi &\text{iff } \quad&\forall t \in \traces(\calK)  \ldot \Pi[\pi \mapsto t] \models_\calK \varphi\\
	\Pi &\models_\calK \exists \pi \ldot \varphi  &\text{iff } \quad&\exists t \in \traces(\calK)  \ldot \Pi[\pi \mapsto t] \models_\calK  \varphi
\end{align*}
We say $\calK$ satisfies $\varphi$, written $\calK \models \varphi$, if $\emptyset \models_\calK \varphi$, where $\emptyset$ denotes the trace assignment with empty domain. 

\section{Game-Based Verification of $\forall^*\exists^*$}\label{sec:gameBased}

Model-checking a HyperLTL formula with $k$ quantifier alternations is $k$-fold exponential \cite{ClarksonFKMRS14,Rabe16}, and complete methods typically utilize expensive operations like automata complementation \cite{FinkbeinerRS15} and inclusion checking \cite{BeutnerF23}.
For $\forall^*\exists^*$ properties, we can soundly (but incompletely) \emph{approximate} the expensive model-checking problem by, instead, constructing a game and searching for a \emph{strategy} that defines witness paths for existentially quantified traces \cite{CoenenFST19,BeutnerF22}.

\subsection{Parity Games}

To model the dynamics of the verification game, we use a turn-based game played between a verifier and a refuter (as done by \cite{BeutnerF22}). 

A parity game (PG) is a tuple $\calG = (V_\veri, V_\refu, v_\mathit{init}, \directions, \allowbreak E, c)$, where $V := V_\veri \uplus V_\refu$ is the set of game vertices, partitioned into vertices controlled by the verifier ($V_\veri$) and refuter ($V_\refu$), $v_\mathit{init} \in V$ is the initial vertex of the game, $\directions$ is a set of directions, $E : V \times \directions \to V$ is the transition function of the game, and $c : V \to \nat$ assigns each state a color used for the parity acceptance condition.

The game is played on an underlying graph whose vertices are controlled by the verifier ($\veri$) or refuter ($\refu$). 
Whenever the game is in a vertex controlled by a player $p \in \{\veri, \refu\}$, the respective player can determine to which vertex the game should progress by choosing some direction from $\directions$.
A \emph{strategy} for the verifier is a function $\sigma : V^*\cdot V_\veri \to \directions$.
The strategy reads a sequence of vertices $v_1\cdots v_n$ (ending in a vertex $v_n \in V_\veri$ controlled by the verifier) and determines a direction $\sigma(v_1\cdots v_n) \in \directions$ in which the game should progress.
A play $\rho \in V^\omega$ is compatible with a strategy $\sigma$ for $\veri$ if  $\rho(0) = v_\mathit{init}$ (i.e., the play starts in $\calG$'s initial vertex), and for every $i \in \nat$, with $\rho(i) \in V_\veri$, we have $\rho(i+1) = E\big(\rho(i), \sigma(\rho[0, i])\big)$.
That is, we construct the play iteratively;
Whenever the game is in a vertex controlled by the verifier, we query strategy $\sigma$ on the current prefix to obtain a direction and update the vertex based on $\calG$'s transition function.
We say a play $\rho \in V^\omega$ is \emph{even} if the minimal color that appears infinitely often on $\rho$ (according to coloring $c$) is even (similar to the acceptance condition used for DPAs).
The verifier wins $\calG$ if there exists a strategy $\sigma$ for the verifier such that every play compatible with $\sigma$ is even.

\subsection{The Verification Game for $\forall^*\exists^*$}

Assume a fixed system $\calK = (S, s_\mathit{init}, \directions, \kappa, \ell)$ and a $\forall\exists$ HyperLTL formula $\forall \pi_1. \exists \pi_2. \psi$.
For our game construction, we represent the temporal requirement expressed by $\psi$ -- the LTL body of $\varphi$ -- as a deterministic automaton. 
Recall that the atomic propositions in $\psi$ are indexed with trace variables, i.e., $\psi$ is an LTL formula over $$\ap_\psi := \big\{ a_\pi \mid a \in \ap, \pi \in  \{\pi_1, \pi_{2}\}  \big\}.$$
We assume that $\calA_\psi = (2^{\ap_\psi}, Q_\psi, q_{0, \psi}, \delta_\psi, c_\psi)$ is a DPA over alphabet $2^{\ap_\psi}$ that recognizes $\psi$, i.e., $\calL(\calA_\psi) = \{  t \in (2^{\ap_\psi})^\omega \mid t \models_\mathit{LTL} \psi \}$ (cf.~\Cref{lem:ltl}).
We can then define a parity game, denoted $\gamee{\calK}{\varphi}$, that captures the iterative construction of traces \cite{BeutnerF22}:

\begin{definition}[$\gamee{\calK}{\varphi}$,\cite{BeutnerF22}]\label{def:pg-const}
	Define the parity game $\gamee{\calK}{\varphi}$ by $\gamee{\calK}{\varphi} := (V_\veri, V_\refu, \allowbreak v_\mathit{init}, \directions, E, c)$, where
	
	\begin{itemize}
		\item $V_\veri := \big\{ \gamenode{s_1,s_2, q, \veri} \mid s_1, s_2 \in S \uplus \{s_\mathit{init}\} \land q \in Q_\psi \big\}$,
		\item $V_\refu := \big\{ \gamenode{s_1, s_2, q, \refu} \mid s_1, s_2 \in S \uplus \{s_\mathit{init}\} \land q \in Q_\psi \big\}$, 
		\item $v_\mathit{init} := \gamenode{s_\mathit{init},  s_\mathit{init}, \allowbreak q_{0, \psi}, \refu}$, 
		\item the set of direction $\directions$ is the same as in $\calK$,
		\item the transition function $E : V \times \directions \to V$ is defined by
		\begin{align*}
			&E\big(\gamenode{s_1,s_2, q, \veri}, d\big) := \big\langle s_1, \kappa(s_2, d), q, \refu \big\rangle\\
			&E\big(\gamenode{s_1,s_2, q, \refu}, d\big) := \\
			&\quad\quad\quad\Big\langle \kappa(s_1, d), s_2, \delta_\psi\Big(q, \bigcup_{i=1}^{2} \big \{ a_{\pi_i} \mid a \in \ell(s_i) \big\} \Big), \veri \Big\rangle,
		\end{align*}
		\item $c\big( \gamenode{s_1, s_2, q, p} \big) := c_\psi(q)$.
	\end{itemize}
\end{definition}

In our game, each vertex $\gamenode{s_1, s_2, q, p}$ tracks a state $s_1$ for $\pi_1$ (called the $\pi_1$-copy), a state $s_2$ for $\pi_2$ (called the $\pi_2$-copy), the current state $q$ of $\calA_\psi$, and the current player $p \in \{\veri, \refu\}$.
In the initial vertex $v_\mathit{init}$, every system copy starts in the initial state, and $\calA_\psi$ begins tracking in its initial state $q_{0, \psi}$.
Intuitively, in each round of the game, the refuter can update the $\pi_1$-copy by moving along some transition in $\calK$, followed by the verifier updating the $\pi_2$-copy; afterward, the game repeats. 
Formally, whenever in a vertex $\gamenode{s_1, s_2, q, \veri}$, the verifier can choose a direction $d \in \directions$, and the $\pi_2$-copy is updated to $\kappa(s_2, d)$ (the first case in the definition of $E$).
Analogously, when in vertex $\gamenode{s_1, s_2, q, \refu}$ the refuter can update the $\pi_1$-copy along some direction (the second case in the definition of $E$). 
When the refuter moves a round of the game has concluded, so we update the state of $\calA_\psi$.
For each $i \in \{1, 2\}$, we thus read of the APs that currently hold in state $s_i$ ($\ell(s_i) \subseteq \ap$) and index all these APs with $\pi_i$ to obtain a letter $\bigcup_{i=1}^{2} \big \{ a_{\pi_i} \mid a \in \ell(s_i) \big\} \in 2^{\ap_\psi}$, which we feed to $\calA_\psi$'s transition function.

As we track separate states for $\pi_1$ and $\pi_2$, every infinite play in $\calG_{\calK, \varphi}^{\forall\exists}$ defines two paths in $\calK$; one for $\pi_1$ (where each step is controlled by the refuter), and one for $\pi_2$ (controlled by the verifier). 
In $\gamee{\calK}{\varphi} $, each vertex $\gamenode{s_1, s_2, q, p}$ is assigned color $c_\psi(q)$ using $\calA_\psi$'s coloring function, so an infinite play in $\gamee{\calK}{\varphi}$ is won by the verifier iff the paths constructed for $\pi_1, \pi_2$ during the gameplay are accepted by $\calA_\psi$ and thus satisfy $\psi$.
Any winning strategy for $\veri$ thus step-wise constructs a witness trace for $\pi_2$, no matter how the refuter constructs $\pi_1$.
It is not hard to see that the existence of a winning strategy for the verifier thus implies that we can always find a witness trace for $\pi_2$ in the HyperLTL semantics:

\begin{restatable}[\!\!\cite{BeutnerF22}]{lemma}{soundnessTheo}\label{theo:soundness-fe}
	If the verifier wins $\gamee{\calK}{\varphi}$, then $\calK \models \varphi$.
\end{restatable}

\begin{remark}\label{rem:usage}
	The game-based approach can be used for automated verification by constructing the game and solving it via an off-the-shelf parity solver (the original motivation of \cite{CoenenFST19,BeutnerF22}).
	However, the appeal of \Cref{theo:soundness-fe} is much broader.
	For example, the user can construct a strategy by using domain knowledge, which enables \emph{interactive} verification even in situations where automated model-checking does not scale (see, e.g., \cite{CorrensonF25}).
	Likewise, checking if a given strategy for the verifier wins $\gamee{\calK}{\varphi}$ is often easier than computing a strategy from scratch; strategies are easy-to-check certificates.
\end{remark}

\section{Game-Based Verification Beyond $\forall^*\exists^*$}\label{sec:game-ii}

The game-based approach from the previous section soundly approximates the semantics of $\forall^*\exists^*$ formulas.
Unfortunately, the approach is limited to $\forall^*\exists^*$ properties and becomes unsound when considering properties beyond $\forall^*\exists^*$.

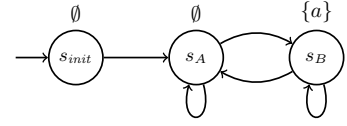
\begin{figure}
	
	\centering
	\scalebox{0.8}{
	\begin{tikzpicture}
		
		\node[circle, draw, thick, label=above:$\emptyset$,minimum size=9mm] at (-2,0) (n0) {\small$s_\mathit{init}$};
		
		\node[circle, draw, thick, label=above:$\emptyset$,minimum size=9mm] at (0,0) (n1) {\small$s_A$};
		
		\node[circle, draw,  thick, label=above:$\{a\}$,minimum size=9mm] at (2,0) (n2) {\small$s_B$};
		
		\draw[->, thick] (n1) to[bend left] (n2);
		\draw[->, thick] (n2) to[bend left] (n1);
		
		\draw[->,  thick] (n1) to[loop below] (n1);
		\draw[->,  thick] (n2) to[loop below] (n2);
		
		\draw[->, thick] (n0) to (n1);
		
		\draw[->, thick] (-3, 0) to (n0);
	\end{tikzpicture}}
	
	\caption{Simple Kripke structure over $\ap = \{a\}$ \vspace{-3mm}}\label{fig:example-unsound}
\end{figure}

\begin{example}\label{ex:unsound}
	Consider the Kripke structure $\calK$ over $\ap = \{a\}$ in \Cref{fig:example-unsound} and the $\exists^1\forall^1$ HyperLTL formula 
	\begin{align*}
		\varphi := \exists \pi_1. \forall \pi_2\ldot (\ltlN\ltlN\ltlN a_{\pi_1}) \leftrightarrow (\ltlN\ltlN a_{\pi_2}),
	\end{align*}
	where the LTL body expresses that AP $a$ should hold in the third step on $\pi_1$ iff it holds in the second step on $\pi_2$. 
	Clearly, $\calK \not\models \varphi$; no matter what \emph{fixed} trace we choose for $\pi_1$, we can always find a trace for $\pi_2$ that violates the LTL body.
	
	Now let us na\"ively adopt the game used in \Cref{sec:gameBased} to this $\exists\forall$ property. 
	That is, we, again, maintain states for all trace variables and let the verifier and refuter iteratively update the existentially and universally quantified system copies, respectively.  
	The resulting game is won by the verifier:
	During the gameplay, the refuter (who controls the state for the universally quantified $\pi_2$) has to decide in the \emph{second} step whether or not AP $a$ should hold (by moving to $s_B$ or $s_A$).
	Only later in the game (in the third round) does the verifier choose if $a$ holds on $\pi_1$. 
	By that time, the verifier can thus react to what the refuter has done in the previous step and set $a$ on $\pi_1$ appropriately, ensuring that $(\ltlN\ltlN\ltlN a_{\pi_1}) \leftrightarrow (\ltlN\ltlN a_{\pi_2})$ is satisfied. 
	The game is won by the verifier, even though the property does not hold. 
\end{example}

\subsection{Multiplayer Games and Partial Information}

In this paper, we propose a novel game-based approximation that applies to arbitrary quantifier structures. 
Our simple yet powerful observation is that the unsoundness is directly linked to the \emph{knowledge} of the player. 
In \Cref{ex:unsound}, the verifier could win the game as it can observe what the refuter did in the previous steps. 
In contrast, in the semantics of an $\exists\forall$ property, any witness for the existential quantifier must be \emph{independent} of the choice for the universal quantifier. 
By identifying knowledge as the core reason for unsoundness, we can extend the game-based approach to \emph{arbitrary} HyperLTL formulas by utilizing \emph{partial information}.
The technical challenge is then to link the knowledge/information of a player to the first-order semantics of HyperLTL.

As our underlying game formalism, we use (sequential) multiplayer parity games \cite{MalvoneMS16}, which extend parity games with multiple players and partial information. 
	A multiplayer parity game under incomplete information (MPG$_\mathit{ii}$) is a tuple 
	$$\calG = (\agents, \{V_p\}_{p \in \agents}, v_\mathit{init}, \directions, E, \allowbreak \{\sim_p\}_{p \in \agents} , c),$$
	where $\agents$ is a finite set of players; for each $p \in \agents$, $V_p$ is a finite set of vertices controlled by $p$. 
	We write $V := \biguplus_{p \in \agents} V_p$ for the set of all vertices in the game (which we assume to be disjoint).
	$v_\mathit{init} \in V$ is the initial vertex of the game, $\directions$ is a set of directions, and $E : V \times \directions \to V$ is the transition function of the game. 
	For each player $p \in \agents$, $\sim_p \subseteq V \times V$ is an equivalence relation on the set of vertices.
	Lastly, $c : V \to \nat$ assigns each vertex a color.

An MPG$_\mathit{ii}$ models the joint behavior of the players in $\agents$, where each player $p \in \agents$ controls a disjoint set of vertices $V_p$.
If the game is currently in a vertex in $V_p$, player $p$ can determine where the game should move next by choosing some direction from $\directions$.
We obtain a standard two-player game (cf.~\Cref{sec:gameBased}) by setting $\agents = \{\veri, \refu\}$. 
Moreover, each player $p \in \agents$ is assigned an indistinguishability relation $\sim_p$, i.e., if $v \sim_p v'$, player $p$ cannot distinguish between $v$ and $v'$.
As usual, we assume that each player is at least able to distinguish whether or not it controls the current vertex.
That is, for every $v \sim_p v'$, we either have $v, v' \in V_p$ or $v, v' \in V \setminus V_p$.

\paragraph*{Strategies and Plays}

Two finite plays $\rho_1, \rho_2 \in V^*$ are indistinguishable for player $p \in \agents$, written $\rho_1 \sim_p \rho_2$, if $|\rho_1| = |\rho_2|$ and for every $0 \leq i < |\rho_1|$, $\rho_1(i) \sim_p \rho_2(i)$.  
A strategy for player $p \in \agents$ is a function $\sigma_p : V^*\cdot V_p \to \directions$, such that $\sigma_p(\rho_1) = \sigma_p(\rho_2)$ whenever $\rho_1 \sim_p \rho_2$. 
The strategy reads a sequence of vertices $v_1\cdots v_n$ (ending in a vertex $v_n \in V_p$) and determines a direction $\sigma_p(v_1\cdots v_n) \in \directions$ in which the game should progress.
This decision must conform to the agent's observations, i.e., if two finite plays appear indistinguishable for $p$, strategy $\sigma_p$ must choose the same direction on both plays.  
Within the game, we obtain infinite plays in $\calG$ by letting strategies for all players interact.
A \emph{global strategy} $\{\sigma_p\}_{p \in \agents}$ assigns each player $p \in \agents$ a strategy $\sigma_p$.
Every global strategy $\{\sigma_p\}_{p \in \agents}$ defines a unique infinite play $\rho \in V^\omega$, where $\rho(0) = v_\mathit{init}$ (i.e., the play starts in $\calG$'s initial vertex), and for every $i \in \nat$, we have $\rho(i+1) = E\big(\rho(i), \sigma_p(\rho[0, i])\big)$ where $p \in \agents$ is the unique player with $\rho(i) \in V_p$. 
That is, whenever the game is currently in a vertex controlled by player $p$, we query $p$'s strategy on the current prefix $\rho[0, i]$ to obtain a direction and update the game's vertex along that direction.
As before, we say a play $\rho$ is \emph{even} if the minimal color that appears infinitely often on $\rho$ is even.
We are interested in checking if a given \emph{coalition} of players $A \subseteq \agents$ can win the game. 
A coalition $A \subseteq \agents$ wins $\calG$, written $\winss{A}{\calG}$, if there exists strategies $\{\sigma_p\}_{p \in A}$ for the players in $A$, such that, no matter what strategies other players use, i.e., for every possible $\{\sigma_p\}_{p \in \agents\setminus A}$, the play resulting from the combined global strategy $\{\sigma_p\}_{p \in \agents}$ is even. 
That is, if the players in $A$ follow their strategies in $\{\sigma_p\}_{p \in A}$, the resulting play satisfies the parity winning condition no matter how the other players behave. 

\subsection{HyperLTL Verification as an MPG$_\mathit{ii}$}

We now present the core contribution of this paper: MPG$_\mathit{ii}$s allow for the sound verification of \emph{arbitrary} HyperLTL formulas. 
For this, assume 
\begin{align*}
	\varphi = \quant_1 \pi_1 \ldots \quant_n \pi_n\ldot \psi
\end{align*}
is a fixed HyperLTL formula over trace variables $\pi_1, \ldots, \pi_n$ with \emph{arbitrary} quantifier structure.
As in \Cref{sec:gameBased}, we assume that $\calA_\psi = (2^{\ap_\psi}, Q_\psi, q_{0, \psi}, \delta_\psi, c_\psi)$ is a DPA over alphabet $\ap_\psi := \big\{a_\pi \mid a \in \ap, \pi \in \{\pi_1, \ldots, \pi_n\} \big\}$ that recognizes $\psi$ (cf.~\Cref{lem:ltl}).

\begin{definition}[$\game{\calK}{\varphi}$]\label{def:construct-pgii}
	Define the MPG$_{\mathit{ii}}$ $\game{\calK}{\varphi}$ by 
	\begin{align*}
		\game{\calK}{\varphi} := (\agents, \{V_p\}_{p \in \agents}, v_\mathit{init}, \directions, E, \{\sim_p\}_{p \in \agents} , c),
	\end{align*}
	where
	\begin{itemize}
		\item $\agents := \{1, \ldots, n\}$,
		\item for each $p \in \agents$,
		\begin{align*}
			V_p &:= \{ \gamenode{s_1, \ldots, s_n, q, p} \mid \\
			&\quad\quad\quad s_1, \ldots, s_n \in S \uplus \{s_\mathit{init}\}, q \in Q_{\psi} \},
		\end{align*}
		\item $v_\mathit{init} := \gamenode{s_\mathit{init}, \ldots,  s_\mathit{init}, \allowbreak q_{0, \psi}, 1}$,
		\item the set of direction $\directions$ is the same as in $\calK$,
		\item the transition function $E : V \times \directions \to V$ is defined by
		\begin{align*}
			&E\Big( \big\langle s_1,\ldots, s_n, q, p \big\rangle, d\Big) := \\
			&\quad\quad\Big\langle s_1, \ldots, s_{p-1}, \kappa(s_p, d), s_{p+1}, \ldots, s_n, q, \mathit{nxt}(p)  \Big\rangle,
		\end{align*}
		for $p > 1$ and for $p = 1$ define
		\begin{align*}
			&E\Big( \big\langle s_1,\ldots, s_n, q, 1 \big\rangle, d\Big) :=  \Big\langle \kappa(s_1, d), s_2, \ldots, s_n, \\
			&\quad\quad\quad\quad\quad\quad \delta_\psi\Big(q, \bigcup_{i=1}^{n} \big \{ a_{\pi_i} \mid a \in \ell(s_i) \big\} \Big), \mathit{nxt}(1) \Big\rangle,
		\end{align*}
		where $\mathit{nxt}(p) := p + 1$ if $p < n$ and $\mathit{nxt}(n) := 1$,
		\item for each $p \in \agents$, 
		\begin{align*}
			\sim_p \, := \,&\Big\{ \big(\gamenode{s_1, \ldots, s_n, q, p'}, \gamenode{s'_1, \ldots, s'_n, q', p''} \big) \mid \\
			&\quad\quad\quad p' = p'' \land  \forall 1 \leq j \leq p \ldot s_j = s_j' \Big\},
		\end{align*}
		\item $c\big( \gamenode{s_1, \ldots, s_n, q, p}\big) := c_\psi(q)$.
	\end{itemize}
\end{definition}

The first key idea is to represent each trace variable in the formula as a separate player, so $\agents = \{1, \ldots, n\}$. 
The vertices in $\calG_{\calK, \varphi}$ are of the form $\gamenode{s_1, \ldots, s_n, q, p}$, where $s_1, \ldots, s_n$ track the current state of the system copies for $\pi_1, \ldots, \pi_n$, respectively,, $q \in Q_\psi$ tracks $\psi$, and $p \in \agents$ determines which player controls this vertex.
Each infinite play in $\calG_{\calK, \varphi}$, therefore, defines $n$ concrete paths (and thus traces) for $\pi_1, \ldots, \pi_n$.
As expected, we start each system in $\calK$'s initial state $s_\mathit{init}$, start $\calA_\psi$ in $q_{0, \psi}$, and let player $1$ begin.
Similar to  \Cref{sec:gameBased}, a vertex $\gamenode{s_1, \ldots, s_n, q, p}$ is assigned color $c_\psi(q)$.
The transition function then allows the players to update their system copy.
Whenever in a vertex $\gamenode{s_1, \ldots, s_n, q, p}$ where $p > 1$ (the first case in the definition of $E$), the direction $d$ (which is chosen by player $p$ controlling this vertex) updates the $\pi_p$-copy by moving along the chosen direction $d \in \directions$ to $\kappa(s_p, d)$.
Afterward, it is the next player's turn: $\mathit{nxt}(p)$ increases $p$ by $1$ and cycles back to player $1$ once the last player (player $n$) has acted.
The players thus take turns updating their system state; first, player $1$ updates the state of the $\pi_1$-copy, then player $2$ updates the state of the $\pi_2$-copy, and so forth, until, finally, player $n$ updates the state of the $\pi_n$-copy, and the game repeats with player $1$.
When it is player $1$'s turn (the second case in the definition of $E$), a game round has just concluded. 
In this case, the direction chosen by player $1$ is used to update the $\pi_1$ copy (similar to the other rounds).
At the same time, we update the automaton state of $\calA_\psi$ by reading the AP evaluation of $s_1, \ldots, s_n$, obtaining a letter $\bigcup_{i=1}^{n} \big \{ a_{\pi_i} \mid a \in \ell(s_i) \big\} \in 2^{\ap_\psi}$.

The last key ingredient is the partial information of each player.
The core problem of the game from \Cref{sec:gameBased} was that the players could observe the global state of the game, leading to unsound behavior.
MPG$_\mathit{ii}$s allow us to precisely determine the information that each player can act on. 
Once we have observed that knowledge is the key to a sound verification game, we can align the player's information with the HyperLTL semantics:
In a HyperLTL formula $\quant \pi_1 \ldots \quant \pi_{i-1}. \exists \pi_i. \quant \pi_{i+1} \ldots \quant \pi_n\ldot \psi$, the choice for $\pi_i$ is only based on the traces $\pi_1, \ldots, \pi_{i-1}$ as those are the traces that are already added to the trace assignment in the semantics of HyperLTL.
We can directly express this in our game definition:
For a player $p$ (that controls $\pi_p$), two vertices $\gamenode{s_1, \ldots, s_n, q, p'}$ and $\gamenode{s_1', \ldots, s_n', q', p''}$ appear indistinguishable if it is the same player's turn ($p' = p''$) and $s_j = s_j'$ for all $j \leq p$, i.e., the state of all traces quantified \emph{before} $\pi_p$ agrees.

\subsection{Soundness}

To use our game as a verification method, we are interested in the strategic ability of all players controlling existentially quantified system copies. 
That is, we define 
$$\agents_\exists := \{ i \in \agents \mid \quant_i = \exists \}.$$
We can then show that our game under partial information constitutes a sound verification approach:

\begin{restatable}{theorem}{sound}\label{theo:soundness}
	If $\winss{\agents_\exists}{\calG_{\calK, \varphi}}$, then $\calK \models \varphi$.
\end{restatable}

\begin{example}
	Consider some formula $\forall \pi_1. \exists \pi_2. \forall \pi_3. \psi$. 
	In our game definition, the player controlling $\pi_2$ can observe the current state of the $\pi_1$-copy and thus react to the behavior of player $1$. 
	However, it cannot base its decision on the behavior of player $3$.
	In the special case of $\exists\pi_1.\forall \pi_2$ properties (like in \Cref{ex:unsound}), player $1$ can only observe its own state. 
	It must thus use the same sequence of directions (and thus define the \emph{same} witness paths), no matter how player $2$ behaves (cf.~\cite{BeutnerF25c}). 
	In \Cref{ex:unsound}, player $1$ would thus lose this game; there is no winning behavior in the third step without observing player $2$'s behavior in the second step.
\end{example}

It is not hard to see that in the special case of $\forall^*\exists^*$ properties, the MPG$_\mathit{ii}$ $\game{\calK}{\varphi}$ coincides with the game $\gamee{\calK}{\varphi}$ from \Cref{sec:gameBased}.

\begin{lemma}\label{lem:same}
	Let $\varphi = \forall \pi_1\ldot \exists \pi_2 \ldot \psi$. Then $\winss{\{2\}}{\calG_{\calK, \varphi}}$ if and only if the verifier wins $\gamee{\calK}{\varphi}$ (cf.~\Cref{sec:gameBased}).
\end{lemma}

Similar to the $\forall^*\exists^*$ game from \Cref{sec:gameBased}, we can use $\game{\calK}{\varphi}$ not only for automated verification but also as a foundation for interactive verification and certificates (cf.~\Cref{rem:usage}).

\subsection{Hierarchical Information}\label{sec:hir}

Multiplayer games under imperfect information are often undecidable \cite{PnueliR90,BerthonMM17}, i.e., there exists no general algorithm to check if a given group of players can win the game. 
Fortunately, there exists a well-known class of imperfect information games that we can solve effectively.
An MPG$_\mathit{ii}$ $(\agents, \{V_p\}_{p \in \agents}, v_\mathit{init}, \directions, E, \{\sim_p\}_{p \in \agents} , c)$ is played under hierarchical information if there exists a total order $\prec$ on $\agents$ such that for every $p' \prec p$, we have $\sim_p \subseteq \sim_{p'}$.
That is, we can order the players such that the information is hierarchical, i.e., player $p$ observes at least as much as all smaller players (w.r.t.~$\prec$). 
By adopting standard techniques, we can show that we can decide if a given group of players can win a game played under hierarchical information \cite{Reif84,peterson2002decision,FinkbeinerS05}.

\begin{restatable}{lemma}{hierarchical}
	The MPG$_\mathit{ii}$ $\calG_{\calK, \varphi}$ from \Cref{def:construct-pgii} is played under hierarchical information.
	Consequently, it is decidable if $\winss{\agents_\exists}{\calG_{\calK, \varphi}}$.
\end{restatable}

\section{Completeness for $\exists^*\forall^*$}\label{sec:ea}

In our game-based view, we let players construct witness traces for existentially quantified traces. 
Compared to the HyperLTL semantics, this limits the power of existential quantification.
For example, in the semantics of a $\forall \pi_1\ldot \exists \pi_2 \ldot \psi$ formula, the choice for $\pi_2$ can be based on the \emph{entire} trace assigned to $\pi_1$. 
In the game-based view, $\pi_2$ is constructed step-wise by a player, so the decision in the $i$th round of the game can only depend on $\pi_1$'s prefix of length $i$. 
This leads to incompleteness, i.e., situations where a property holds, but the game is not won by $\agents_\exists$. 

\begin{example}[\!\!\cite{BeutnerF22}]\label{ex:incomplete}
	Consider the Kripke structure $\calK$ in \Cref{fig:example-unsound} and the HyperLTL formula $\varphi = \forall \pi_1\ldot \exists \pi_2\ldot  \ltlG\ltlF(a_{\pi_2} \leftrightarrow \ltlN a_{\pi_1})$ which requires that $\pi_2$ predicts the next value of $\pi_1$ infinitely often. Clearly, $\calK \models \varphi$, but $\neg \winss{\{2\}}{\calG_{\calK, \varphi}}$: In $\calG_{\calK, \varphi}$, player $2$ has to decide in each step whether AP $a$ should hold but does not know what player $1$ will do on trace $\pi_1$ in the future. 
	No matter what player $2$ does, player $1$ can thus always ensure that $a_{\pi_2} \not\leftrightarrow \ltlN a_{\pi_1}$.
\end{example}

Our game is thus incomplete; already on $\forall^*\exists^*$ properties where our game coincides with the full-information games from \cite{BeutnerF22}, cf.~\Cref{lem:same}.
While incomplete for $\forall^*\exists^*$, our game is, perhaps surprisingly, complete for $\exists^*\forall^*$ properties.

\begin{restatable}{theorem}{complete}
	Assume $\varphi$ is a $\exists^*\forall^*$ HyperLTL formula. 
	Then $\winss{\agents_\exists}{\calG_{\calK, \varphi}}$ if and only if $\calK \models \varphi$. 
\end{restatable}

Note that we can check any $\forall^*\exists^*$ property by checking the negated property (which is $\exists^*\forall^*$). 
Our game-based approach, therefore, constitutes a sound-and-complete model-checking technique for all HyperLTL formulas with \emph{at most one} quantifier alternation. 
The cost of this completeness manifests itself in the additional complexity; our game-based approach for $\forall^*\exists^*$ properties (which coincides with \cite{BeutnerF22}) is incomplete but results in a standard game under full information.
If we, instead, check the negated $\exists^*\forall^*$ formula, our game is complete, but the game uses partial information, making automated game-solving more challenging. 

\section{Prophecy Variables}\label{sec:comp}

For properties beyond $\exists^*\forall^*$, we need additional tools to counteract the incompleteness.
In this section, we study the use of prophecy variables in our game-based framework; a technique originally studied in \cite{CoenenFST19,BeutnerF22} for the verification of $\forall^*\exists^*$ hyperproperties (building on the seminal work by Abadi and Lamport \cite{AbadiL88}).
At a high level, a prophecy provides limited information about the future behavior of other players.
For example, in the setting of a $\forall \pi_1\ldot \exists \pi_2\ldot \psi$ formula, the player controlling $\pi_2$ only observes the past behavior of the player controlling $\pi_1$.
For such a formula, a prophecy is an LTL formula $\xi$ over trace variable $\pi_1$ (cf.~\cite{BeutnerF22}). 
In each step of the game, the player controlling $\pi_2$ can then query an oracle that tells them whether or not the future behavior of $\pi_1$ will satisfy $\xi$, and base its decision on the additional information provided by the oracle. 

\begin{example}
	In \Cref{ex:incomplete}, the player controlling $\pi_2$ does not have a winning strategy as it does not know if $a$ will hold on $\pi_1$ in the \emph{next} step.
	In this example, it suffices for the player controlling $\pi_2$ to have access to an oracle that, in each step, predicts whether prophecy $\xi := \ltlN a_{\pi_1}$ holds (cf.~\cite{BeutnerF22}).
	If $\xi$ holds (so $\ltlN a_{\pi_1}$), the player can move the $\pi_2$-copy to state $s_B$ (where $a$ holds, cf.~\Cref{fig:example-unsound}), thus ensuring that $a_{\pi_2} \leftrightarrow \ltlN a_{\pi_1}$.
\end{example}

Prophecies essentially combat the problem of having \emph{too little} information about the \emph{future}, which already leads to \emph{incompleteness} in the $\forall^*\exists^*$ setting. 
The main contribution of this paper is that we can use partial information to avoid players having \emph{too much} information about the behavior of (certain) other players, which would lead to \emph{unsoundness}.
In this section, we extend the prophecy framework of \cite{BeutnerF22} from $\forall^*\exists^*$ to arbitrary quantifier prefixes.
Our prophecy construction combines two ideas: our observation model ensures that players do not observe too much (to ensure soundness), while prophecies provide missing information about future events.

\subsection{Prophecies and Partial Information}\label{sub:prophecies}

To keep notation simple, we assume for the remainder of this section -- and w.l.o.g. -- that the hyperproperty in question is of the form $\varphi = \forall \pi_1\ldot \exists \pi_2\ldot \forall \pi_3\ldots \forall \pi_{2n-1} \ldot \exists \pi_{2n}\ldot\psi$, i.e., strictly alternates between universal (at odd indices) and existential (at even indices) quantification.

\begin{example}\label{ex:runningExample}
	Consider the formula 
	\begin{align*}
		&\varphi := \forall \pi_1\ldot \exists \pi_2 \ldot \forall \pi_3 \ldot \exists \pi_4\ldot \ltlG \ltlF (a_{\pi_2} \leftrightarrow \ltlN a_{\pi_1}) \land \\
		&\quad\quad\quad\ltlN\ltlN \big( a_{\pi_4} \leftrightarrow \ltlG(a_{\pi_1} \leftrightarrow a_{\pi_2} \land a_{\pi_2} \leftrightarrow a_{\pi_3}) \big).
	\end{align*}
	It requires $\pi_2$ to globally predict the next step of $\pi_1$, and $\pi_4$ should, in the second step, predict whether $\pi_1$, $\pi_2$, and $\pi_3$ agree on $a$.
	The KS $\calK$ in \Cref{fig:example-unsound} satisfies $\varphi$, yet coalition $\{2,4\}$ loses $\calG_{\calK, \varphi}$.
	To win this game, the player controlling $\pi_2$ needs (in every step) information about the future of $\pi_1$, and the player controlling $\pi_4$ needs (in the second step) information about the joint future of $\pi_1, \pi_2$, and $\pi_3$.
\end{example}

\subsection{Prophecies and Prophecy Variables}

To ensure soundness, we need to ensure that each player $i$ is (via the prophecies) only given information over traces quantified \emph{before} $\pi_i$.
For each existentially quantified trace $\pi_{2i}$, we therefore track a separate set of prophecies:

\begin{definition}\label{def:family}
	A \emph{prophecy family} is a collection $\vec{\Xi} = \{\Xi_{2i-1}\}_{i=1}^n$, where $\Xi_{2i-1}$ is a finite set of LTL formulas over trace variables from $\{\pi_1, \ldots, \pi_{2i-1}\}$.
\end{definition}

Intuitively, $\Xi_{2i-1}$ contains all prophecies that provide information to the player controlling trace $\pi_{2i}$.
Consequently, the formulas in $\Xi_{2i-1}$ only reason about traces $\{\pi_1, \ldots, \pi_{2i-1}\}$, which are exactly the traces player $i$ can observe in the game. 

\begin{example}\label{ex:runningExample2}
	Consider the property in \Cref{ex:runningExample}.
	We can construct the prophecy family $\{\Xi_1, \Xi_3  \}$, where $\Xi_1 := \{ \ltlN a_{\pi_1}  \}$, and $\Xi_3 := \{ \ltlG(a_{\pi_1} \leftrightarrow a_{\pi_2} \land a_{\pi_2} \leftrightarrow a_{\pi_3})  \}$. 
	These prophecies provide exactly the information needed by players $2$ and $4$.
	That is, if players $2$ and $4$ could, \emph{in each step of the game}, determine if the prophecies in $\Xi_1$ and $\Xi_3$ hold, respectively, they can construct appropriate witness traces and win $\calG_{\calK, \varphi}$. 
\end{example}

Now assume we have a fixed family of LTL formulas $\vec{\Xi} = \{\Xi_{2i-1}\}_{i=1}^n$. 
The intuition behind the prophecies is that each player $\pi_{2i}$ is provided with an oracle that -- in each step of the game-- tells her which of the formulas in $\Xi_{2i-1}$ hold. 
Following \cite{BeutnerF22}, we will use \emph{prophecy variables} to formalize this oracle. 
A prophecy variable is essentially an AP that we add to the system, and we ensure that the value of this variable (AP) reflects the truth value of the prophecy formula.
For this, we assume that $P$ is a set of prophecy variables for $\vec{\Xi}$, i.e., for each $1 \leq i \leq n$ and $\xi \in \Xi_{2i-1}$, there exists a corresponding prophecy variable $p^\xi \in P$.
A player can thus query an oracle on whether prophecy $\xi$ currently holds, by simply reading the prophecy variable (AP) $p^\xi$.
The key idea now is that we can attach these prophecy variables to \emph{universally} quantified traces \cite{BeutnerF22}.
That is, we let the opposing players $\agents \setminus \agents_\exists$ (controlling the universally quantified traces) determine the truth value of the prophecy variables, and then ensure, within the HyperLTL formula, that the prophecy variables are set correctly: 
That is, we modify the HyperLTL formula to ensure that $\xi \in \Xi_{2i-1}$ holds iff the prophecy variable $p^\xi$ is set to true on trace $\pi_{2i-1}$.
This ensures that the player $2i$ controlling $\pi_{2i}$ can query the prophecies in $\Xi_{2i-1}$ by looking at the current state of $\pi_{2i-1}$ (and the Boolean value of the prophecy variables, i.e., AP,  in that state), but the players controlling $\pi_1, \ldots, \pi_{2i-2}$ cannot.

As a first step, we add the variables in $P$ to the system, which can  be set arbitrarily in each step:

\begin{definition}[$\calK^{P}$]\label{def:addPToK}
	Given a KS $\calK = (S, s_\mathit{init}, \directions, \kappa, \ell)$ over $\ap$ and a disjoint set of prophecy variables $P$ ($\ap \cap P = \emptyset$), define the modified KS $\calK^{P} := (S \times 2^P, s_\mathit{init}, \directions \times 2^P, \kappa^P, \ell^P)$ over $\ap \uplus P$ where for each direction $(d, A) \in \directions \times 2^P$, we define $\kappa^P$ by
	\begin{align*}
		\kappa^P(s_\mathit{init}, (d, A)) &:= (\kappa(s_\mathit{init}, d), A)\\
		\kappa^P((s, A'), (d, A)) &:= (\kappa(s, d), A)
	\end{align*}
	and $\ell^P(s_\mathit{init}) = \ell(s_\mathit{init})$ and $\ell^P(s, A) := \ell(s) \cup A$.
\end{definition}
Intuitively, $\calK^{P}$ generates all traces of $\calK$ extended with all possible evaluations on the prophecy variables.

We then modify the body of the HyperLTL formula such that the original property is only required to hold if all prophecies are set correctly, i.e., a prophecy variable in $p^\xi$ is set to true iff the future traces satisfy $\xi$. 

\begin{definition}[$\varphi^{{P}, \vec{\Xi}}$]\label{def:addProphToForm}
	Given a prophecy family $\vec{\Xi} = \{\Xi_{2i-1}\}_{i=1}^n$ and a corresponding set of prophecy variables $P$, define the modified HyperLTL formula $\varphi^{{P}, \vec{\Xi}}$ by 
	\begin{align*}
		\varphi^{{P}, \vec{\Xi}} := \forall \pi_1\ldot &\exists \pi_2\ldot \forall \pi_3\ldots \forall \pi_{2n-1} \ldot \exists \pi_{2n}\ldot \\
		&\bigg(\ltlN\ltlG \bigwedge_{i=1}^n \Big(\bigwedge_{\xi \in \Xi_{2i-1}} \big((p^\xi)_{\pi_{2i-1}} \leftrightarrow \xi\big)\Big)\bigg)\to \psi.
	\end{align*}
\end{definition}

That is, we require $\psi$ (the original LTL body of $\varphi$) to hold, if in every step, for every $1 \leq i \leq n$, and every prophecy formula $\xi \in \Xi_{2i-1}$, $\xi$ holds iff the corresponding prophecy variable $p^\xi$ holds on trace $\pi_{2i-1}$.
Note how $\varphi^{P, \vec{\Xi}}$ connects the prophecy variables with the underlying prophecy.
If some prophecy variable $p^\xi$ for $\xi \in \Xi_{2i-1}$ is set on (the universally quantified) trace $\pi_{2i-1}$, the player $2i \in \agents_\exists$ controlling the existentially quantified trace $\pi_{2i}$ can assume that $\xi$ holds, and vice versa.
If this is not the case, i.e., player $2i-1$ sets $p^\xi$ incorrectly, the premise of $\varphi^{P, \vec{\Xi}}$ is violated, so the LTL body of $\varphi^{P, \vec{\Xi}}$ is vacuously true, and $\agents_\exists$ wins the game in \Cref{def:construct-pgii}. 
Note that trace $\pi_{2i-1}$ is universally quantified, so we consider all possible valuations of the prophecy variables, including the unique valuation where the prophecy variables are set correctly in each step, i.e., $\ltlN \ltlG \bigwedge_{\xi \in \Xi_{2i-1}} \big((p^\xi)_{\pi_{2i-1}} \leftrightarrow \xi\big)$.
Observe that we only require the prophecies to be set correctly after the first step (using a single $\ltlN$), as the unique initial state $s_\mathit{init}$ of a KS does not record prophecy variables. 

\subsection{Prophecies and Games}
We can use the information provided via prophecies in our game-based approach.
Instead of checking if the players in $\agents_\exists$ win $\game{\calK}{\varphi}$, we can then check if they win $\calG_{\calK^{P}, \varphi^{P, \vec{\Xi}}}$. 
In the latter game, the additional premise in $\varphi^{P, \vec{\Xi}}$ ensures that the players can use prophecies to peek at future moves of the universal traces.

\begin{example}\label{ex:runningExample3}
	We consider the example from \Cref{ex:runningExample,ex:runningExample2}.
	We already argued that $\agents_\exists = \{2, 4\}$ loses $\game{\calK}{\varphi}$.
	Now consider the prophecy family $\vec{\Xi}$ from \Cref{ex:runningExample2}, with corresponding prophecy variables $P := \{ \{p\}, \{pp\} \}$.
	Using \Cref{def:addProphToForm}, we construct
	\begin{align*}
		&\varphi^{P, \vec{\Xi}} = \forall \pi_1\ldot \exists \pi_2 \ldot \forall \pi_3 \ldot \exists \pi_4\ldot 
		\bigg[ \ltlN\ltlG \big(p_{\pi_1} \leftrightarrow \ltlN a_{\pi_1}\big) \land \\
		&\quad\quad\quad\ltlN\ltlG \big(pp_{\pi_3} \leftrightarrow \ltlG(a_{\pi_1} \leftrightarrow a_{\pi_2} \land a_{\pi_2} \leftrightarrow a_{\pi_3})\big) \bigg] \to \psi,
	\end{align*}
	where $\psi$ is the original LTL body (cf.~\Cref{ex:runningExample}).
	It is easy to see that $\{2, 4\}$ wins $\calG_{\calK^{P}, \varphi^{P, \vec{\Xi}}}$:
	the prophecy variable $p$ on $\pi_1$ hints at the next move of $\pi_1$. 
	If, for example, player $1$ sets $p$ to true, player $2$ can assume that $\ltlN a_{\pi_1}$ holds (if it does not, the premise of $\varphi^{P, \vec{\Xi}}$'s LTL body is violated, and so the play is trivially won by $\{2, 4\}$).
	Likewise, the prophecy variable $pp$ provides the necessary information for player $4$. 
\end{example}

\subsection{Soundness}

The key result we are left to prove is that the addition of prophecies does not change the HyperLTL semantics, even though the LTL body of $\varphi^{P, \vec{\Xi}}$ is weaker than the body of $\varphi$.

\begin{restatable}{theorem}{proph}\label{theo:introProph}
	Assume a prophecy family $\vec{\Xi} = \{\Xi_{2i-1}\}_{i=1}^n$ and a corresponding set of prophecy variables $P$.
	Then $\calK \models \varphi$ if and only if $\calK^{P} \models \varphi^{P, \vec{\Xi}}$.
\end{restatable}

\Cref{theo:introProph} allows us to soundly combine prophecies with the game-based approach:
If $\winss{\agents_\exists}{\calG_{\calK^{P}, \varphi^{P, \vec{\Xi}}}}$, then, by \Cref{theo:soundness}, we have $\calK^{P} \models \varphi^{P, \vec{\Xi}}$, so, by \Cref{theo:introProph}, we have $\calK \models \varphi$.
Prophecy variables thus constitute a tool that can strengthen game-based verification in the presence of arbitrary quantifier alternations.
This is particularly relevant when using our approach as an interactive proof technique.
Once a suitable set of prophecies is found (i.e., a family $\vec{\Xi}$ s.t.~$\agents_\exists$ wins $\calG_{\calK^{P}, \varphi^{P, \vec{\Xi}}}$), the prophecies, together with the winning strategies for $\agents_\exists$, are an easy-to-check certificate of satisfaction.

\section{Related Work}\label{sec:relatedWork}

\paragraph*{Logics for Hyperproperties}

Most logics for expressing temporal hyperproperties use HyperLTL-style quantification over execution traces \cite{CoenenFHH19,GutsfeldMO20,Rabe16,NguyenKJDJ17}.
In such logics, quantifier alternations are frequently used to, e.g., reason about non-determinism in the system.
Our paper proposes a principled approach to deal with quantifier alternations that can be easily extended to other logics that feature HyperLTL-style quantification over system paths/traces.

\paragraph*{HyperLTL Verification}

Finite-state model-checking of HyperLTL is decidable \cite{ClarksonFKMRS14}, and complete algorithms rely on expensive automata complementation or language inclusion checks \cite{FinkbeinerRS15,BeutnerF23}.
To approximate this expensive problem, Hsu et al.~\cite{HsuSB21,HsuSSB23} propose a bounded model-checking approach for HyperLTL by unrolling the system and property into a QBF formula. 
The other prominent approximation for HyperLTL is the game-based approach \cite{CoenenFST19,BeutnerF22}, which forms the foundation of the present paper. 
Both approximations are orthogonal to each other. 
In the game-based approach, we use strategies to resolve existential quantification and can thus reason about temporal behavior along \emph{infinite} paths.
In contrast, the QBF-based encoding features the same first-order semantics used in HyperLTL, but bounds the length of the traces, limiting the approach to properties that can be verified or refuted within a bounded timeframe.

\paragraph*{Advantages of Game-based Verification}
The game-based approach has multiple advantages over automata-complementation-based methods. 
Firstly, it allows verification in settings where complementation-based approaches fail. 
For example, the game-based approach can be used to verify \emph{infinite-state} systems by constructing abstract games \cite{BeutnerF22b,ItzhakySV24}, or utilizing infinite-state game solvers \cite{battigalli2003rationalizability,AlfaroHM01,HeimD24,FaellaP23,BaierCFFJS21}.
Secondly, the game-based approach allows for \emph{interactive proofs}, i.e., the user can manually construct a proof by incrementally defining a winning strategy \cite{CorrensonF25}.
This facilitates proofs in situations where automated methods do not scale.
And lastly, the game-based approach naturally yields \emph{certificates} of satisfaction, i.e., the computed winning strategy (combined with a set of prophecies, if needed) can easily be checked by independent strategy checkers \cite{BeutnerFG24}. 

In this paper, we propose an extension of the game-based approach to arbitrary quantifier structures, based on the key conceptual contribution of leveraging \emph{partial information}, extending our earlier extended abstract \cite{BeutnerF25a}.
Based on this key conceptual idea, we provide a sound framework that allows us to utilize the benefits of game-based verification for arbitrary quantifier structures.
While solving games under partial information is expensive, our paper also allows for cheaper approaches that build upon the game-based interpretation.
For example, we can attempt to find \emph{positional} strategies for $\agents_\exists$.
Finding positional strategies is much cheaper than finding strategies under perfect recall (i.e., unbounded memory), and, if winning positional strategies for $\agents_\exists$ are found, our results allow us to soundly conclude that $\calK \models \varphi$.  
This creates a spectrum of techniques with varying complexity and expressiveness (i.e., in some instances, positional strategies suffice, in others, we might need bounded memory or even perfect recall).
In the case of $\exists^*\forall^*$ properties, our definition yields a game where the player acts without any information, a setting explored extensively within the planning community (cf.~\emph{conformant planning}) \cite{GoldmanB96,BeutnerF25c}.

\paragraph*{Solving Games Under Incomplete Information}

We believe that the primary use case of our approach lies in its ability to construct sound-by-design interactive proofs and certificates.
Nevertheless, if paired with a solver for games under partial information, our approach could underpin a fully-automated verification pipeline.
The study of (multiplayer) games under incomplete information has a long tradition, mostly relying on using a powerset construction or lattice framework to track belief states \cite{Reif84,peterson2002decision,ChatterjeeD10,RaskinCDH07,BerwangerCWDH10,WulfDR06}.
Strategy-based logic can explicitly reason about the strategic abilities of players, and extensions to incomplete information exist \cite{BerthonMMRV21,BeutnerF24a,BerthonMM17,PileckiBJ14,BeutnerF25b}. 
Many of the frameworks used to study incomplete information can be used in our setting:
for example, partially observable non-deterministic (POND) planning reasons about an agent that acts in a non-deterministic environment, essentially defining a game under partial information \cite{BrafmanSZ13}, which we can use for hyperproperty verification \cite{BeutnerF24b,BeutnerF25c}.
Likewise, multi-agent planning \cite{GmytrasiewiczD05,GaleslootSJ024}, partially observable Markov decision processes with multiple agents \cite{ZhangL11,AmatoO15}, or multi-agent reinforcement learning \cite{BusoniuBS08} study strategy synthesis in partially observable domains.

\section{Conclusion and Future Work}

In this work, we have proposed a novel hyperproperty verification method using games. 
In contrast to previous game-based approaches, our method is not limited to $\forall^*\exists^*$ properties but soundly approximates \emph{arbitrary} quantifier structures. 
Moreover, we designed a prophecy mechanism that aligns with the partial information of the player.
Our work creates numerous avenues for future work, both in theory and practice.
In theory, it is interesting to study the expressive power of our game-based approximation when using prophecies. 
Prophecies are a complete proof technique in the setting of $\forall^*\exists^*$ properties, i.e., whenever a property holds, there exists some finite set of $\omega$-regular (not necessarily LTL-definable) prophecies such that the game is won by the verifier \cite{BeutnerF22}. 
The high-level idea of this construction is to construct prophecies that directly determine successor states, i.e., prophecies $\xi_{s, s'}$ that hold iff, when in state $s$, moving to $s'$ is the ``optimal'' move for the verifier (see \cite{BeutnerF22} for details). 
We conjecture that completeness also holds in the presence of arbitrary quantifier alternations:
For a property $\quant_1 \pi_1 \ldots \quant_{i-1} \pi_{i-1} \exists \pi_i. \quant_{i+1} \pi_{i+1} \ldots \quant_n \pi_n\ldot \psi$, we can design prophecies that precisely define the optimal behavior for player $i$ (for a fixed system $\calK$, $\quant_{i+1} \pi_{i+1} \ldots \quant_n \pi_n\ldot \psi$ is just an $\omega$-regular property over $\pi_1, \ldots, \pi_{i}$).
In practice, we can utilize our games as an interactive proof technique, e.g., using the coinductive framework of \cite{CorrensonF25}.

\section*{Acknowledgments}

This work was supported by the European Research Council (ERC) Grant HYPER (101055412), and by the German Research Foundation (DFG) as part of TRR 248 (389792660).

%\balance
\bibliographystyle{IEEEtran}
\bibliography{references}

\iffullversion

\appendices

\section{Soundness}\label{app:sound}

\sound*
\begin{proof}
	As all traces quantified in $\varphi$ originate from some path in $\calK$, we can reason about paths in $\paths(\calK)$ instead of traces in $\traces(\calK)$.
	To show that $\calK \models \varphi$, we now construct a \emph{Skolem function} $f_p : \paths(\calK)^{p-1} \to \paths(\calK)$ for each $p \in \agents_\exists$; a standard proof technique for first-order logic \cite{shoenfield2018mathematical}.
	Intuitively, $f_p$ provides a concrete witness path for $\pi_p$ when given the $p-1$ paths used for $\pi_1, \ldots, \pi_{p-1}$ \cite{shoenfield2018mathematical}.
	We want to find a family of Skolem functions $\{f_p\}_{p \in \agents_\exists}$ such that for every path combination $(\tau_1, \ldots, \tau_n) \in \paths(\calK)^n$ where $\tau_p = f_p(\tau_1, \ldots, \tau_{p-1})$ whenever $p \in \agents_\exists$, $(\tau_1, \ldots, \tau_n)$ satisfies $\psi$ (the LTL body of $\varphi$).
	That is, for all existentially quantified paths, we use the Skolem function but choose arbitrary paths for all universally quantified ones.
	If we can find such a family, we can conclude that $\calK \models \varphi$ \cite{shoenfield2018mathematical}.
	
	Let $\{\sigma_p\}_{p \in \agents_\exists}$ be a joint strategy that wins $\calG_{\calK, \varphi}$ for coalition $\agents_\exists$, which exists as we assume  $\winss{\agents_\exists}{\calG_{\calK, \varphi}}$.
	Each $\sigma_p$ is a function $V^* \to \directions$ such that $\sigma_p(\rho_1) = \sigma_p(\rho_2)$ whenever $\rho_1 \sim_p \rho_2$.
	The key proof idea is now to observe that -- due to the observation relation -- we can view $\sigma_p$ as a function that reads the quotient of $V$ by $\sim_p$, i.e., a function $\sigma_p : (V / \!\!\! \sim_p)^* \to \directions$.
	As player $p$ can only observe the first $p$ states in each vertex (by definition of $\sim_p$), we can thus view $\sigma_p$ as a function $\sigma_p : (\Delta_p )^* \to \directions$, where $\Delta_p = (V /  \!\!\! \sim_p) := \{\langle s_1, \ldots, s_p \rangle \mid s_1, \ldots, s_p \in S \}$.
	We can then use strategy $\sigma_p$ to construct Skolem function $f_p : \paths(\calK)^{p-1} \to \paths(\calK)$ by querying it on \emph{prefixes}. 
	The Skolem function $f_p$ will see $p-1$ infinite paths in $\calK$.
	In contrast, $\sigma_p$ observes the states of the first $p-1$ paths, i.e., has only seen a prefix of the final paths.
	We will use $\sigma_p$ to construct $f_p$ by iteratively querying it on those prefixes. 
	For this, let $\tau_1, \ldots, \tau_{p-1} \in \paths(\calK)$ and we want to define $f_p(\tau_1, \ldots, \tau_{p-1})$.
	We use $\tau_1, \ldots, \tau_{p-1} \in \paths(\calK)$ to construct sequences $T_0, T_1, \ldots \in (\Delta_p)^*$, where $T_i$ is the prefix of the game at which strategy $\sigma_p$ will select the $i$th direction, which we can easily construct inductively from prefixes $\tau_1[0,i], \ldots, \tau_{p-1}[0, i]$ and the first $i-1$ directions chosen by $\sigma_p$.
	Now define $f_p(\tau_1, \ldots, \tau_{p-i})$ as the unique path $\tau$ such that $\tau(0) = s_\mathit{init}$, and $\tau(i+1) = \kappa(\tau(i), \sigma_p(T_i))$, i.e., we use the directions chosen by $\sigma_p$ on each of the prefixes $T_0, T_1, \ldots$ to construct $\tau$.
	That is, even though we already know the entire paths $\tau_1, \ldots, \tau_{p-1}$, we act as if we only knew the prefix up to the current round of the game.
	Note how the incomplete information is necessary to enable this construction; The Skolem function $f_p$ is only given the $p-1$ paths quantified before $\pi_p$, so we do \emph{not} know the paths for $\pi_{p+1}, \ldots, \pi_n$. 
	As we can view $\sigma_p$ as a function of $\Delta_p$, the $p-1$ paths for $\pi_1, \ldots, \pi_{p-1}$ suffice for the simulation.

	It remains to argue that the family of Skolem functions $\{f_p\}_{p \in \agents_\exists}$ witnesses $\calK \models \varphi$. 
	For this, let $(\tau_1, \ldots, \tau_n)$ be a path combination under $\{f_p\}_{p \in \agents_\exists}$.
	We need to show that $(\tau_1, \ldots, \tau_n)$ satisfies $\psi$.
	Here, the key observation is that all path combinations permitted by $\{f_p\}_{p \in \agents_\exists}$ directly correspond to plays in $\calG_{\calK, \varphi}$ under strategies $\{\sigma_p\}_{p \in \agents_\exists}$ (as we have used $\sigma_p$'s direction to construct existentially quantified paths).
	In each infinite play in $\calG_{\calK, \varphi}$, the DPA $\calA_\psi$ tracks if the simulated paths satisfy $\psi$, and -- as $\{\sigma_p\}_{p \in \agents_\exists}$ wins the game -- this automaton accepts every play under $\{\sigma_p\}_{p \in \agents_\exists}$.
	As this holds for any combination $(\tau_1, \ldots, \tau_n)$ under $\{f_p\}_{p \in \agents_\exists}$, we can conclude that $\calK \models \varphi$ as required. 
\end{proof}

\hierarchical*
\begin{proof}
	The player in $\calG_{\calK, \varphi}$ are $\{1, \ldots, n\}$ and we order them $1 \prec 2 \prec \cdots \prec n$.
	It is easy to see that player $p$ can observe more than any player $p' < p$. 
	If two vertices $\gamenode{s_1, \ldots, s_n, q, x}, \gamenode{s_1', \ldots, s_n', q', x'}$ are indistinguishable for $p$ (i.e., $\gamenode{s_1, \ldots, s_n, q, x} \sim_p \gamenode{s_1', \ldots, s_n', q', x'}$) , we have $x = x'$ and $\forall 1 \leq i \leq p\ldot s_i = s_i'$. 
	Clearly, this implies that $\forall 1 \leq i \leq p'\ldot s_i = s_i'$ (as $p' < p$), so $\gamenode{s_1, \ldots, s_n, q, x} \sim_{p'} \gamenode{s_1', \ldots, s_n', q', x'}$.
\end{proof}

\section{Completeness for $\exists^*\forall^*$}

\complete*
\begin{proof}
	The first direction follows from \Cref{theo:soundness}. 
	For the other direction assume $\varphi = \exists \pi_1 \ldots \pi_n\ldot \forall \pi_{n+1} \ldots \pi_{n+m}\ldot 
	\psi$ is the $\exists^*\forall^*$ formula and $\calK \models \varphi$. 
	As $\calK \models \varphi$, we get witness traces $t_1, \ldots, t_n \in \traces(\calK)$ such that $[\pi_1 \mapsto t_1, \ldots, \pi_n \mapsto t_n] \models \forall \pi_{n+1} \ldots \pi_{n+m}\ldot \psi$. 
	Let $\tau_1, \ldots, \tau_n \in \paths(\calK)$ be paths that generate $t_1, \ldots, t_n$, respectively.  
	Now, for each $1 \leq p \leq n$, let $\sigma_p$ be the strategy that simply generates path $\tau_p$, i.e., in the $j$th game round, it picks a direction to move the $\pi_p$-copy to state $\tau_p(j)$. 
	It is easy to see that $\{\sigma_i\}_{i=1}^n$ is a winning strategy for $\agents_\exists = \{1, \ldots, n\}$; the strategy generates traces $t_1, \ldots, t_n$, so, no matter how the opponents play, the resulting paths satisfy $\psi$. So $\winss{\agents_\exists}{\calG_{\calK, \varphi}}$ as required.
\end{proof}

\section{Prophecies}

\proph*
\begin{proof}
	For the first direction, assume $\calK \models \varphi$. Recall that $\varphi = \forall \pi_1\ldot \exists \pi_2\ldot \forall \pi_3\ldots \forall \pi_{2n-1} \ldot \exists \pi_{2n}\ldot\psi$
	For each $1 \leq i \leq n$, let $f_{2i}$ be a Skolem function for the $i$th existential quantifier, i.e., $f_{2i} : \traces(\calK)^{2i-1} \to \traces(\calK)$ maps any combination of $2i-1$ traces for the previous quantified trace variables $\pi_1, \ldots, \pi_{2i-1}$ to a valid choice for $\pi_{2i}$.
	We want to show that $\calK^{P} \models \varphi^{P, \vec{\Xi}}$. 
	For this, we construct Skolem functions $\tilde{f}_{2i} : \traces(\calK^{P})^{2i-1} \to \traces(\calK^{P})$ for $1 \leq i \leq n$ for each existential quantifier in $\varphi^{P, \vec{\Xi}}$.
	To define, $\tilde{f}_{2i}$, let $(t_1, \ldots, t_{2i-1}) \in \traces(\calK^{P})^{2i-1}$.
	Let $(t_1', \ldots, t_{2i-1}') \in \traces(\calK)^{2i-1}$ be the corresponding traces in $\calK$ obtained by projecting on the original APs, i.e., removing all prophecy variables added in \Cref{def:addPToK}.
	Let $t := f_{2i}(t_1', \ldots, t_{2i-1}') $ be the witness assigned to this combination by the Skolem function witnessing  $\calK \models \varphi$. 
	We define $\tilde{f}_{2i}(t_1, \ldots, t_{2i-1})$ as an arbitrary extension of $t$ with prophecy variables.
	No matter what traces we choose for universally quantified traces, the Skolem functions $\tilde{f}_2, \ldots, \tilde{f}_{2n}$ yield witness traces such that the $2n$ combined traces $t_1, \ldots, t_{2n}$ satisfy $\psi$ (the LTL body of $\varphi$) as ${f}_2, \ldots, {f}_{2n}$ show $\calK \models \varphi$.
	Any trace combination that satisfies $\psi$  also satisfies the LTL body of $\varphi^{P, \vec{\Xi}}$ (as the conclusion of the implication is $\psi$). 
	We thus get $\calK^{P} \models \varphi^{P, \vec{\Xi}}$ as required.
	
	For the reverse direction, assume $\calK^{P} \models \varphi^{P, \vec{\Xi}}$. For $1 \leq i \leq n$, let $f_{2i} : \traces(\calK^{P})^{2i-1} \to \traces(\calK^{P})$ be a Skolem function that shows this. 
	As before, we construct Skolem functions $\tilde{f}_{2i} : \traces(\calK)^{2i-1} \to \traces(\calK)$ that show $\calK \models \varphi$. 
	The construction of this Skolem function is more involved, as the LTL body of $\varphi^{P, \vec{\Xi}}$ is weaker than $\psi$, i.e., a combination of traces might satisfy the body $\varphi^{P, \vec{\Xi}}$ without satisfying $\psi$ (the LTL body of $\varphi$). 
	The key proof idea is that the premise of $\varphi^{P, \vec{\Xi}}$ only  refers to prophecy variables on \emph{universally} quantified traces.  As the semantics of HyperLTL consider all possible traces for universally quantified traces and $\calK^P$ generates all possible prophecy combinations, there will also exists some trace where all prophecy variables are set correctly (i.e., a prophecy variable $p^\xi $  is set some step iff the $\xi$ holds on the future execution). 
	If we always use a universally quantified trace where prophecies are set correctly, we ensure that the premise of $\varphi^{P, \vec{\Xi}}$ is satisfied, which implies that  $\psi$ is also satisfied (the conclusion of $\varphi^{P, \vec{\Xi}}$).
	In our construction, we thus always pick traces for the universally quantified trace variables that satisfy the premise of $\varphi^{P, \vec{\Xi}}$.
	That is, we pick traces such that the prophecy variables are set correctly, as stated in the premise of $\varphi^{P, \vec{\Xi}}$'s body (cf.~\Cref{def:addProphToForm}).
	To define $\tilde{f}_{2i} : \traces(\calK)^{2i-1} \to \traces(\calK)$, let $t_1, \ldots, t_{2i-1} \in \traces(\calK)$. We extend these traces into traces $t'_1, \ldots, t'_{2i-1} \in \traces(\calK^{P})$ by setting the prophecy variables in $P$. 
	We do so iteratively; for existentially quantified traces, we can consider an arbitrary extension (prophecy variables on existentially quantified traces are never referenced in the LTL body of $\varphi^{P, \vec{\Xi}}$), for universally quantified traces (i.e., those traces where prophecy variables are referenced), we set a prophecy variable iff the corresponding prophecy formula holds in that step.
	Concretely:  We define $t'_1$ by setting a prophecy variable $p^\xi \in P_1$ in step $j$ iff $t_1[j, \infty]$ satisfies $\xi$ (recall that $\xi$ is a formula over $\{\pi_1\}$, cf.~\Cref{def:family}). 
	This way, we ensure $\ltlG \bigwedge_{\xi \in \Xi_{1}} ((p^\xi)_{\pi_{1}} \leftrightarrow \xi)$, i.e., all prophecy variables on $\pi_1$ are set correctly (cf.~\Cref{def:addProphToForm}).
	The prophecies on $t'_2$ (which is existentially quantified) can be set arbitrarily; they are never referenced in $\varphi^{P, \vec{\Xi}}$'s body (cf.~\Cref{def:addProphToForm}).
	In the next step, we define $t'_3$ by setting a prophecy variable $p^\xi \in P_3$ in step $j$ iff $t_1[j, \infty]$, $t_2[j, \infty]$, and $t_3[j, \infty]$, together, satisfies $\xi$ (recall that $\xi$ is a formula over $\{\pi_1, \pi_2, \pi_3\}$, cf.~\Cref{def:family}). 
	We continue this for $t'_5, \ldots, t'_{2i-1}$.
	After having defined $t'_1, \ldots, t'_{2i-1}$, we can compute $t := f_{2i}(t'_1, \ldots, t'_{2i-1}) \in \traces(\calK^{P})$ (i.e., use the original Skolem function showing $\calK^{P} \models \varphi^{P, \vec{\Xi}}$ on our extended traces). We define $\tilde{f}_{2i} (t_1, \ldots, t_{2i-1}) \in \traces(\calK)$ by projecting $t$ on the original APs, i.e., remove all prophecy variables. 
	When using Skolem function $\tilde{f}_2, \ldots, \tilde{f}_{2n}$ to resolve existentially quantified traces in $\varphi$, we always obtain traces that satisfy $\psi$: In the construction, we picked the evaluation of prophecy variables such that premise of $\varphi^{P, \vec{\Xi}}$'s LTL body is satisfied. 
	As $\{f_{2i}\}_{i=1}^n$ witnesses  $\calK^{P} \models \varphi^{P, \vec{\Xi}}$, all combinations that satisfy the premise also satisfy the conclusion of $\varphi^{P, \vec{\Xi}}$'s LTL body, which is exactly $\psi$.
	The Skolem functions $\{\tilde{f}_{2i}\}_{i=1}^n$ thus show $\calK \models \varphi$ as required.
\end{proof}

\fi

\end{document}